\def\year{2022}\relax
\documentclass[letterpaper]{article} 
\usepackage{aaai22}  
\usepackage{times}  
\usepackage{helvet}  
\usepackage{courier}  
\usepackage[hyphens]{url}  
\usepackage{graphicx} 
\usepackage{paralist}
\usepackage{amsmath}
\usepackage{amsfonts}
\usepackage{amssymb}
\usepackage{comment}
\usepackage{amsthm}
\usepackage{nicefrac}
\usepackage{dsfont}
\usepackage{subcaption}
\usepackage{color}
\usepackage{mathtools}
\usepackage{natbib}  
\usepackage{caption} 
\DeclareCaptionStyle{ruled}{labelfont=normalfont,labelsep=colon,strut=off} 
\usepackage{algorithm}
\usepackage{algorithmic}

\usepackage{newfloat}
\usepackage{listings}
\floatstyle{ruled}
\newfloat{listing}{tb}{lst}{}
\floatname{listing}{Listing}

\DeclareMathOperator*{\argmax}{arg\,max}
\DeclareMathOperator*{\argmin}{arg\,min}
\DeclareMathOperator*{\proj}{proj}

\newtheorem{theorem}{Theorem}[]
\newtheorem{example}{Example}[]
\newtheorem{proposition}[theorem]{Proposition}
\newtheorem{definition}{Definition}[]

\newcommand\cj[1]{\footnote{{\color{blue} Comment J: #1}}}

\definecolor{burgundy}{rgb}{0.8, 0.0, 0.13}
\definecolor{burgund}{rgb}{0.9, 0.13, 0.0}

\newcommand{\Alt}{\mathcal{X}}

\title{Truth-tracking via Approval Voting: Size Matters}
\author {
    Tahar Allouche,\textsuperscript{\rm 1}
    Jérôme Lang, \textsuperscript{\rm 1}
    Florian Yger \textsuperscript{\rm 1}
}
\affiliations {
    \textsuperscript{\rm 1}  LAMSADE, CNRS, PSL, Université Paris-Dauphine\\
    tahar.allouche@dauphine.eu, lang@lamsade.dauphine.fr, florian.yger@lamsade.dauphine.fr
}

\begin{document}

\maketitle

\begin{abstract}
Epistemic social choice aims at unveiling a hidden ground truth given votes, which are interpreted as   noisy signals about it. We consider here a simple setting where votes consist of approval ballots: each voter approves a set of alternatives which they believe can possibly be the ground truth. Based on the intuitive idea that more reliable votes contain fewer alternatives, we define several noise models that are approval voting variants of the Mallows model. The likelihood-maximizing alternative is then characterized as the winner of a weighted approval rule, where the weight of a ballot decreases with its cardinality. We have conducted an experiment on three image annotation datasets; they conclude that rules based on our noise model outperform standard approval voting; the best performance is obtained by a variant of the Condorcet noise model. 
\end{abstract}

\section{Introduction}

Epistemic social choice deals with the problem of unveiling a hidden \emph{ground truth} state from a set of some possible states, given the reports of some voters. Votes are seen as \emph{noisy} reports on the ground truth. The distribution of these reports is modelled by a noise model, sometimes tuned by some parameter reflecting the competence (expertise, reliability) of these voters.

The space of frameworks for epistemic social choice varies along two 
dimensions: the nature of the ground truth and the format of the reports (ballots expressed by voters). Depending on the framework chosen, the ground truth may be a single alternative, a set of alternative or a ranking over alternatives. {\em We assume the simplest form of ground truth: it is a single alternative} (the unique correct answer). Still depending on the framework, ballots may also contain a single alternative, a set of alternatives, or a ranking over alternatives.  {\em We assume that they are subsets of alternatives, that is, approval ballots.} Requiring voters to give only one answer (that is, a single alternative) is often too constraining because voters may be uncertain and believe that several alternatives may possibly be the ground truth. This is the path followed by \cite{isapproval2015,approval2015,learning2017}.

While some classical anonymous rules have been shown to be optimal under some assumptions, the aggregation rule may, when possible, assign different weights to the voters according to their expertise.
Whilst this is doable if we have additional information about voter expertise or when we keep a record of their answers to different questions, estimating the individual expertise gets complicated when we have no prior information about voters and when the sole information we have are votes on a single issue. This leads to the \emph{single-question wisdom of the crowd problem} for which the seminal work \cite{solution2017} proposes a novel solution, namely selecting the alternative which is \emph{surprisingly popular}. Although it proved to be an efficient way to get around the problem of estimating the voters' reliabilities, its major drawback is that it requires the elicitation of further information: each voter has to report her answer and her beliefs about the answers of the remaining voters.

Now we suggest that there is an alternative approach that does not require this surplus of information and that 
simply relies on the truthfulness of voters.
\cite{approval2015} have defined a proper mechanism to incentivize the participants to select an alternative if and only if they believe it can be the winning one. An intuitive idea might be to consider that {\em smaller ballots, \textit{i.e.} answers that contain less alternatives, are more reliable}: a voter who knows the true answer (or, more generally, who believes to know it) will probably select only one alternative and a voter who selects all alternatives has probably no idea whatsoever of the correct answer. For instance,
if voters hear a speech and are asked to detect the language in which it is spoken, we may give more weight to a voter approving \textit{Arabic} and \textit{Hebrew} than to one approving \textit{Arabic}, \textit{Hebrew}, \textit{Persian} and \textit{Turkish}. 

Based on this intuition, more weight must be assigned to smaller ballots. Rules that work this way, which we call \emph{size-decreasing approval rules}, are part of the family of \emph{size approval rules} \cite{size2009}. 
Our goal is to motivate the use of such rules from an epistemic social choice point of view. To this purpose, we will study a family of noise models which are approval voting variants of the Mallows model, and prove that in many cases the optimal rule is \emph{size-decreasing}.

The paper is structured as follows. In Section \ref{sec: related} we discuss related work. In Section \ref{sec: framework} we define the framework and the family of noise models we consider. Section \ref{sec: anonymous} characterizes all anonymous noises whose associated optimal rule is size-decreasing. In Section \ref{sec: non-anonymous}, we consider a more general noise where voters have different noise parameters,
prove that under some mild assumption, \emph{the expected number of selected alternatives grows when the voter is less reliable}, and then give an explicit expression for the expected size of the ballot as a function of the reliability parameter of a voter for a Condorcet-like noise model. Section \ref{sec: experiments} focuses on real datasets on which first we test the hypothesis that smaller ballots are more reliable then we apply different size-decreasing rules associated to various noise models to assess their performances. Section \ref{sec: conclusion} concludes.

\section{Related Work}\label{sec: related}
\paragraph{Epistemic social choice}
It 
studies how a ground truth can be recovered from noisy votes, viewing voting rules as maximum likelihood estimators. It dates back from \cite{condorcet1785} and has lead to a lot of developments in the last 30 years. Condorcet's celebrated {\em jury theorem} considers $n$ independent, equally reliable voters and two alternatives that are {\em a priori} equally likely, and  states that if every voter votes for the correct alternative with probability $p>\frac{1}{2}$, then the majority rule outputs the correct decision with a probability that increases with the number of voters and tends to 1 when the latter grows to infinity. See \cite{collective2017} and \cite{Premises2008} for proofs and discussion.

The framework was later generalized to more than two alternatives 
\cite{condorcet1988}, to voters with different competences  \cite{ShapleyGrofman84,maximum2004}, to a nonuniform prior over alternatives \cite{Ben-YasharN97,optimal2001}, 
to various noise models \cite{common2005,ConitzerRX09}, to correlated votes \cite{voting2011,epistemic2017}, to multi-issue domains \cite{XiaCL10} and to multi-winner voting \cite{Evaluating2020}. 
\citet{truth2019} define a method to aggregate votes weighted according to their average proximity to the other votes as an estimation of their reliability. A review of the field can be found  in 
\cite{ElkindSlinko16}. 

Epistemic voting with approval ballots has scarcely been considered. \cite{isapproval2015} study noise models for which approval voting is optimal given $k$-approval votes, in the sense that the objectively best alternative gets elected, the ground truth being a ranking over all alternatives.
\cite{learning2017} prove that the number of samples needed to recover the ground truth ranking over alternatives with high enough probability from approval ballots is exponential if ballots are required to approve $k$ candidates, but polynomial if the size of the ballots is randomized. 

\paragraph{Crowdsourcing}

\cite{Axiomatic2014,Empirical2014} give a social choice-theoretic study of collective annotation tasks. \cite{ShahZ20} design mechanisms for incentive-compatible elicitation with approval ballots in crowdsourcing applications. \cite{solution2017} introduce the \emph{surprisingly popular} approach to solve the single-question problem.
The approach was successfully generalized to the case where the ground truth is a ranking \cite{HosseiniM0S21}.

Beyond social choice, collective annotation has also been studied in the machine learning community.
\cite{maximum1979} used an expectation-maximization (EM) approach 
for retrieving true binary labels. This approach has been improved along with other methods namely in \cite{learning2010,multidimensional2010,Minimax2017,domain2019}.

\section{Framework}\label{sec: framework}
Consider a set $N$ of $n$ voters and a set of $m \geq 2$ alternatives $\Alt=\{a_1\dots,a_m\}$.
 The (hidden) ground truth consists of a single alternative $a^*$. Voters cast approval ballots $A_i \subseteq \Alt$ consisting of their noisy estimates of the ground truth. Voters who approve no alternative or all alternatives do not bring any meaningful information, therefore without loss of generality, 
 we assume that for all $i$, $A_i \neq \emptyset$ and $A_i \neq \Alt$.

All along this paper, we will model the distribution of these approval ballots by approval voting variants of the \emph{Mallows} noise model. The Mallows distribution was originally defined on \emph{rankings}: we adapt it to subsets of alternatives, keeping the idea that the probability of a subset decreases as its distance from a central point increases, the dispersion being modelled by a parameter $\phi$.

In general, we will call an approval Mallows
noise model any model where voters' ballots are independent (we keep this hypothesis all along the paper) and there exist $n$ parameters $\phi_i \in (0,+\infty)$ and a function $d:\Alt \times \mathcal{P}(\Alt) \mapsto \mathds{R}$ such that  and for any voter $i \in N$:
$$ P_{\phi_i,d}( A_i | a^*=a) = \frac{1}{\beta_i} \phi_i^{d(a^*,A_i)},  \forall a \in \Alt$$
where $\beta_i$ is the corresponding normalization factor. If $\phi_i = \phi$ for all $i \in N$, we say the model is {\em anonymous}.

In the remaining of the paper we will only focus on \emph{neutral} noise models. The neutrality of a noise is defined as its invariance by any permutation of the alternatives:
$$\forall \pi \in \sigma(\Alt), P_{\phi,d}(A|a^*=a)=P_{\phi,d}(\pi(A)|a^*=\pi(a)) $$
We can immediately see that a Mallows noise is neutral if and only if its associated function $d$ is neutral (invariant by a permutation of the alternatives). 

A noise model is neutral if 
$d(a,A)$ depends only on 
$|a \cap A|$ (that is, 1 if $a \in A$ and 0 if $a \notin A$)\footnote{We omit the curly brackets and write $a \cap A$ for $\{a\} \cap A$.} and $|A|$:
\begin{proposition} 
A noise model associated to a function $d$ is neutral if and only if there exists a unique  function $\psi_d:\{0,1\} \times \{0,\dots,m\} \setminus (1,0)\mapsto \mathds{R}$, such that $d(a,A)= \psi_d(|a\cap A|,|A|)$.\footnote{ $(1,0)$ is excluded from the domain of $\psi_d$ simply because ($|A|=0$ and $|a \cap A| = 1$) is impossible.}
\end{proposition}
\begin{proof}
If $d(a,A)= \psi_d(|a\cap A|,|A|)$ then since $\psi_d$ is neutral, $d$ is neutral. 
Conversely, assume $d$ is neutral. We claim that for any two pairs $(a,A)$ and $(b,B)$ such that $(|a\cap A|,|A|) = (|b\cap B|,|B|)$ we have $d(a,A)= d(b,B)$. Assume first $a \in A$ (and therefore, $b \in B$). Consider a permutation $\pi$ such that $\pi(a) = b$ and $\pi(A) = B$. Then $d(b,B) = d(\pi(a), \pi(A)) = d(a,A)$. The argument for the case $ a \notin A$ (and $b \notin B$) is similar. Thus, $d(a,A)$ depends only on $|a\cap A|$ and $|A|$, which means that there is a function $\psi_d:\{0,1\} \times \{0,\dots,m\} \setminus (1,0)\mapsto \mathds{R}$ such that $d(a,A)= \psi_d(|a\cap A|,|A|)$. Uniqueness is immediate.
\end{proof}

\begin{example}
For the Hamming distance we have that: $$d_H(a,A)=|\overline{a}\cap A|+|a \cap \overline{A}|=1-2|a\cap A|+|A|$$
so $d_H(a,A)=\psi_{d_H}(|a\cap A|,|A|)$,
where $\psi_{d_H}(t,j)=1-2t+j$.
$d_H(a,A)$ takes its minimal value 0 for $A = \{a\}$ and its maximal value $m$ for $A =\Alt \setminus \{a\}$.

For the Jaccard distance \cite{jaccard1901}
we have $$d_J(a,A)= \frac{|\overline{a}\cap A|+|a \cap \overline{A}|}{|\overline{a}\cap A|+|a \cap \overline{A}|+|a \cap A|} $$ so $d_J(a,A)=\psi_{d_J}(|a\cap A|,|A|)$ where $\psi_{d_J}(t,j)=1-\frac{t}{1-t+j}$.
$d_J(a,A)$ takes (again) its minimal value 0 for $A = \{a\}$ and its maximal value $1$ for $A =\Alt \setminus \{a\}$.
\end{example}

\section{Anonymous Noise and Size-decreasing Approval Rules}\label{sec: anonymous}
In this section, we suppose that voters share a common (unknown) noise parameter $\phi \in (0,1)$ and that there exists some function $d:\Alt \times \mathcal{P}(\Alt) \mapsto \mathds{R}$ and its associated function $\psi_d$ such that, for any $a\in \Alt$:
$$P_{\phi,d}( A_i | a^*=a) =\frac{1}{\beta} \phi^{d(a^* , A_i)}= \frac{1}{\beta} \phi^{\psi_d(|a^* \cap A_i|,|A_i|)} $$


After formally defining the notion of size-decreasing rules, we state the main result of this section which characterizes all the Mallows anonymous noises (that is, all functions $d$) whose associated maximum likelihood rule is size-decreasing. We will see that this is the case for some natural functions $d$, that we will test later on in the experiments.

\begin{definition}[Size Approval Rule]
Consider a function
$$
    \begin{array}{cccl}
        v :& \mathcal{P}(\Alt)^n &\longrightarrow &\Alt  \\
         &  (A_1,\dots,A_n) &\mapsto & v(A_1,\dots,A_n)
        \end{array}
$$
that, for each approval profile $A=(A_1,\dots,A_n)$, assigns a winning alternative $v(A_1,\dots,A_n) \in \Alt$. We say that $v$ is a size approval rule
 if there exists a vector $w=(w_0,\dots,w_m) \in \mathds{R}^{m+1}$ such that:
$$v(A_1,\dots,A_n) = \argmax_{a \in \Alt} app_w(a) $$
where $app_w$ is the weighted approval score defined by:
$$app_w(a) = \sum_{i: a\in A_i} w_{|A_i|} $$ 

A size approval rule $v$ is {\em size-decreasing}
if its associated vector $w = (w_0, \dots, w_m) \in \mathds{R}^{m+1}$ is such that $w_j>w_{j+1}$ for all $1\leq j\leq m-2$.
\end{definition}




\begin{example}
The size approval rule associated to the vector of weights given by $w_j=n^{m-j}$
is size-decreasing in the most extreme sense, as it is 
{\em lexicographic}: it outputs the alternative which appears most often in singleton ballots,  in case of ties it considers ballots of size $2$ and so on.
\end{example}

\begin{definition}[Maximum Likelihood Estimation Rule]
We define the function: 
$$
    \begin{array}{cccl}
        \zeta_d :& \mathcal{P}(\Alt)^n &\longrightarrow &\Alt  \\
         &  (A_1,\dots,A_n) &\mapsto & \argmax\limits_{a \in \Alt} P_d(A_1,\dots,A_n|a^*=a)
        \end{array}
$$
which given an approval profile outputs the maximum likelihood estimator of the ground truth alternative.
\end{definition}


The next theorem aims to characterize the functions $d$ for which the maximum likelihood estimation rule $\zeta_d$ is a size-decreasing approval rule.

\begin{theorem}\label{charac_d}
For $n \geq 3$,
the maximum likelihood estimation rule $\zeta_d$ is a size-decreasing approval rule if and only if $\Delta \psi_d: j \mapsto \psi_d(0,j)-\psi_d(1,j)$ is decreasing. 
\end{theorem}
\begin{proof}
First, for any approval profile $A=(A_1,\dots,A_n)$,
\begin{small}
\begin{align*}
    \zeta_d(A) & = \argmax_{a \in \Alt} \prod_{i=1}^n \frac{1}{\beta} \phi^{d(a,A_i)}
     = \argmin_{a \in \Alt} \sum_{i=1}^n d(a,A_i)\\
     &= \argmin_{a \in \Alt} \sum_{i=1}^n \psi_d(|a \cap A_i|,|A_i|)\\
    & = \argmin_{a \in \Alt} \sum_{i:a \in A_i} \psi_d(1,|A_i|)+\sum_{i:a \notin A_i} \psi_d(0,|A_i|)\\
    & = \argmin_{a \in \Alt} \underbrace{\sum_{i=1}^n \psi_d(0,|A_i|)}_{Cste}-\underbrace{\sum_{i:a \in A_i} \psi_d(0,|A_i|)-\psi_d(1,|A_i|)}_{\Delta \psi_d(|A_i|)}\\
    & = \argmax_{a \in \Alt} \sum_{i:a \in A_i} \Delta \psi_d(|A_i|)
\end{align*}
\end{small}
\noindent$\impliedby$: If $\Delta \psi_d$ is decreasing then we can immediately prove that $\zeta_d$ is a size-decreasing approval rule with a weight vector $w$ such that $w_j = \Delta \psi_d(j)$ for any $1\leq j \leq m-1$.

\noindent$\implies$: Suppose that $\zeta_d$ is a size-decreasing approval rule. Thus, there exists a weight vector $w = (w_0, \dots, w_m) \in \mathds{R}^{m+1}$ such that $w_j > w_{j+1}$ for $1\leq j\leq m-2$ and for any approval profile $A=(A_1,\dots,A_n)$ we have: $$\zeta_d(A)=\argmax_{a \in \Alt} \sum_{i:a \in A_i} w_{|A_i|}$$
If $m=3$:
Let $\Alt=\{a,b,c\}$. To prove that $\Delta \psi_d(1)> \Delta \psi_d(2)$ consider the following approval profile:
$$
\left\{
    \begin{array}{cll}
        A_1 & = \{a\} & \\
        A_i & = \{a,b\} & ,\forall i \in \{2, \dots, n-1\}\\
        A_n & =\{b,c\}
        \end{array}
\right.
$$
Which yields the following weighted approval scores:
$$
\left\{
    \begin{array}{cl}
        app_w(a) & = w_1 + (n-2) w_2  \\
        app_w(b) & = (n-1) w_2\\
        app_w(c) & = w_2
        \end{array}
\right.
$$
Since $w_1>w_2$ we have $\zeta_d(A)=a$, which implies that $\argmax_{e \in \Alt} \sum_{i:e \in A_i} \Delta \psi_d(|A_i|)=a$. In particular:
$$\sum_{i:a \in A_i} \Delta \psi_d(|A_i|) > \sum_{i:b \in A_i} \Delta \psi_d(|A_i|)  $$
So\\
\centerline{$\Delta \psi_d(1) +(n-2)\Delta \psi_d(2) > (n-1)\Delta \psi_d(2)$}\\
which implies that $\Delta \psi_d(1) > \Delta \psi_d(2)$.

If $m>3$: Let $\Alt=\{a,b,c,e_1,\dots,e_{m-3}\}$. To prove that  $\Delta \psi_d(1)> \Delta \psi_d(2)$ we use the same approval profile as above. To prove that $\Delta \psi_d(j)> \Delta \psi_d(j+1)$ for $j\geq 2$, consider the following approval profile:
$$
\left\{
    \begin{array}{cll}
        A_1 & = \{a,e_1,\dots,e_{j-1}\} & \\
        A_i & = \{a,b\} & ,\forall i \in \{2, \dots, n-1\}\\
        A_n & =\{b,c,e_1,\dots,e_{j-1}\}
        \end{array}
\right.
$$
Which yields the following weighted approval scores:
$$
\left\{
    \begin{array}{cll}
        app_w(a) & = w_j \mbox{~~~~}+ (n-2) w_2 &  \\
        app_w(b) & = w_{j+1} + (n-2) w_2 &\\
        app_w(c) & = w_{j+1}  & \\
        app_w(e_l) & = w_j + w_{j+1} & \forall l \in \{1,\dots,j-1\}
        \end{array}
\right.
$$
Since $w_2>w_j>w_{j+1}$ we have $\zeta_d(A)=a$, which implies that $\argmax_{e \in \Alt} \sum_{i:e \in A_i} \Delta \psi_d(|A_i|)=a$. In particular, we have
$$\sum_{i:a \in A_i} \Delta \psi_d(|A_i|) > \sum_{i:b \in A_i} \Delta \psi_d(|A_i|)  $$
So:
$$\Delta \psi_d(j) +(n-2)\Delta \psi_d(2) > \Delta \psi_d(j+1)+(n-2)\Delta \psi_d(2) $$
which implies that: $$\Delta \psi_d(j) > \Delta \psi_d(j+1)$$
\end{proof}
\paragraph{Interpretation:} Consider an anonymous noise $P_{\phi,d}$, where $d$ is such that $\Delta \psi_d$ is decreasing.  Now consider any alternative $a \in \Alt$, and for any $k \in [1,m-2]$, let $A_k,A_{k+1},B_k,B_{k+1}$ be four sets such that $a \in A_k \cap A_{k+1}$ and $a \notin B_k \cup B_{k+1}$ and $|A_k|=|B_k|=k $ and $|A_{k+1}|=|B_{k+1}|=k+1 $. We can easily check that since $\phi \in (0,1)$ and $d(a,B_k)-d(a,A_{k})<d(a,B_{k+1})-d(a,A_{k+1})$, we have 
$$\frac{P_{\phi,d}(B_k|a^*=a)}{P_{\phi,d}(A_k|a^*=a)} < \frac{P_{\phi,d}(B_{k+1}|a^*=a)}{P_{\phi,d}(A_{k+1}|a^*=a)} $$
which implies the following:
\begin{compactitem}
    \item If it is more likely that a voter casts a $k$-sized ballot not containing the ground truth than a $k$-sized ballot  that contains it, then it is even more likely that she casts a $(k+1)$-sized ballot not containing the ground truth than a $(k+1)$-sized ballot  that contains it.
    \item If it is more likely that a voter casts a $(k+1)$-sized ballot containing the ground truth than a $(k+1)$-sized ballot  that does not, then it is even more likely that she casts a $k$-sized ballot containing the ground truth than a $k$-sized ballot  that does not.
\end{compactitem}

We now give some examples with some usual functions $d$. We will see that the maximum likelihood estimation rule associated to the Jaccard distance is size-decreasing with weights $w_{|A|}=\frac{1}{|A|}$, and that the maximum likelihood estimation rule associated to the Hamming distance is not size-decreasing.
\begin{example}\label{ex: Jacc}
Consider the Jaccard distance given by:
$$d_J(a,A)=\psi_{d_J}(|a\cap A|,|A|)=1-\frac{|a \cap A|}{|A|-|a \cap A|+1} $$
which gives:
$$\Delta \psi_{d_J}(j) = \psi_{d_J}(0,j)-\psi_{d_J}(1,j)=\nicefrac{1}{j}  $$
By Theorem \ref{charac_d}, we conclude that the maximum likelihood estimation rule $\zeta_{d_J}$ is a size-decreasing approval rule with weights
$w_j = \nicefrac{1}{j}$.
\end{example}

\begin{example}
Consider the Hamming distance given by:
$$d_H(a,A)= |a \cap \overline{A}|+|\overline{a}\cap A|= 1+|A|-2|a\cap A| $$
Which gives us that:
$$\Delta \psi_{d_H}(j) = \psi_{d_H}(0,j)-\psi_{d_H}(1,j)=2  $$
Therefore, the maximum likelihood estimation rule $\zeta_{d_{H}}$ is a size approval rule with constant weights: it is the {\em standard approval rule} (SAV), that selects the alternative with the maximum number of approvals. It can be seen immediately that SAV is not size-decreasing; however, 
it is, so to say, size-non-increasing, and thus can be seen as the limit of size-decreasing rules.
\end{example}

As a consequence of Theorem \ref{charac_d}, we can easily prove that, for an anonymous 
noise, the maximum likelihood estimation rule associated to a function $d$  defined as a linear combination of the quantities $|a\cap A|$ and $|A|$ is not a size-decreasing rule (this is the case for the Hamming distance). More generally, this applies to any function $d$ such that $\psi_d$ can be additively separated into two terms $\psi_d(|a\cap A|,|A|)=f(|a\cap A|)+g(|A|)$. In the next section, we will consider this particular family of separable functions with a non-anonymous 
noise, where each voter has her own noise parameter $\phi_i$.

\section{Non-anonymous Separable Noise }\label{sec: non-anonymous}
\subsection{The General Case}
Consider a set of $m$ alternatives $\Alt=\{a_1, \dots,a_m\}$ and a ground truth answer $a^*$.
 Consider also a neutral function $d:\Alt\times \mathcal{P}(\Alt)\longrightarrow \mathds{R}$ with an associated function $\psi_d:\{0,1\} \times \{0,\dots,m\} \setminus (1,0)\mapsto \mathds{R}$ which can be separated into two quantities:
 $$\psi_d(|a\cap A|,|A|) = f(|a\cap A|)+g(|A|) $$
 We define a non-anonymous Mallows noise model, where for each voter $i \in N$ there exists a parameter $\phi_i \in ]0,+\infty[$ such that, for any $a \in \Alt$:\footnote{Recall that voters cannot cast empty or full approval ballots. Therefore we suppose that $P(\emptyset|a^*=a)=P(\Alt|a^*=a)=0$.}
$$P_{\phi_i,d}(A_i|a^*=a) = \frac{1}{\beta_i} \phi_i^{f(|a^* \cap A_i|)+g(|A_i|)}$$
Notice that in this case, a bigger individual noise parameter $\phi_i$ models a less reliable voter (her distribution is less condensed around the ground truth). The aim of the next result is to motivate the use of size-decreasing approval rules to aggregate approvals generated from such distributions. More precisely, the goal is to find sufficient conditions on $f$ and $g$ that makes the expected size of the voter's ballot $\mathds{E}_{\phi,d}[|A_i|]$ grows as the voter become less reliable (\textit{i.e.} as her noise parameter $\phi_i$ grows).

We will denote $f(1)=f_1, f(0)=f_0$ and $\Delta f=f_0 - f_1$ which would naturally be positive $\Delta f>0$. We will also denote by $\Delta g_{k,t}=g(k)-g(t)$.

\begin{theorem}
If for every $1\leq t < k \leq m-1$ we have that:
$$\Delta g_{k,t} \geq \frac{k-t}{2} \Delta f $$
Then:
$$\frac{\partial \mathds{E}_{\phi,d}[|A_i|]}{\partial \phi} \geq 0$$
\end{theorem}

\begin{proof}
We will just give a sketch of the proof. The full version is included in the appendix.
For any prior distribution on $a^*$, it suffices to prove that for any $a \in \Alt$ we have that: $$ \frac{\partial \mathds{E}_{\phi,d}[|A_i||a^*=a]}{\partial \phi} \geq 0$$
since $\mathds{E}_{\phi,d}[|A_i|] = \sum\limits_{a\in \Alt}  \mathds{E}_{\phi,d}[|A_i|| a^*=a] P(a^*=a)  $.

\allowdisplaybreaks
Let $a \in \Alt$. First we can compute the normalization factor:
$$\beta = \sum_{t=1}^{m-1} {m-1 \choose t-1} \phi^{f_1+g(t)}+ {m-1 \choose t} \phi^{f_0+g(t)} = \sum_{t=1}^{m-1} h_t(\phi) $$
where $h_t(\phi) = {m-1 \choose t-1} \phi^{f_1+g(t)}+ {m-1 \choose t} \phi^{f_0+g(t)} $.
So:
\begin{small}
$$\mathds{E}_{\phi,d}[|A_i| | a^*=a] = \sum_{t=1}^{m-1} t P(|A_i|=t|a^*=a)= \frac{\sum\limits_{t=1}^{m-1} t\times h_t(\phi)}{\sum\limits_{t=1}^{m-1} h_t(\phi)}$$
\end{small}
Thus the derivative reads:
\begin{small}
\begin{align*}
\frac{\partial \mathds{E}_{\phi,d}[|A_i||a]}{\partial \phi} & = \frac{\sum\limits_{t=1}^{m-1} t h'_t(\phi)\sum\limits_{k=1}^{m-1}
    h_k(\phi)-\sum\limits_{t=1}^{m-1} t h_t(\phi)\sum\limits_{k=1}^{m-1} h'_k(\phi)}{\left( \sum\limits_{t=1}^{m-1}  h_t(\phi)\right)^2}\\
    & \propto \sum_{t=1}^{m-2}\sum_{k=t+1}^{m-1} (k-t) \left[h'_k(\phi) h_t(\phi)-h'_t(\phi)h_k(\phi) \right]
\end{align*}
\end{small}
We can already notice that to guarantee that $\frac{\partial \mathds{E}_{\phi,d}[|A_i|| a^*=a]}{\partial \phi} \geq 0$ it suffices that $\Delta h_{k,t}(\phi)=\left[h'_k(\phi) h_t(\phi)-h'_t(\phi)h_k(\phi) \right] \geq 0$ for all $1 \leq t < k\leq m-1$.

We have that:
\begin{align*}
    h'_k(\phi) h_t(\phi) & = \left[(g(k)+f_1){m-1 \choose k-1} \phi^{g(k)+f_1-1}\right.\\
    &+ \left.(g(k)+f_0){m-1 \choose k} \phi^{g(k)+f_0-1} \right] \\
    & \times \left[ {m-1 \choose t-1} \phi^{f_1+g(t)}+ {m-1 \choose t} \phi^{f_0+g(t)} \right]\\
\end{align*}
So we can check that:
\begin{small}
$$\Delta h_{k,t}(\phi) \propto \left[r(\phi)+\frac{m-k}{k}+ \frac{m-t}{t}\right] \Delta g_{k,t} - \left(\frac{m-t}{t}- \frac{m-k}{k}\right) \Delta f $$
\end{small}
where
$ r(\phi) =  \phi^{-\Delta f}+\frac{m-k}{k} \frac{m-t}{t} \phi^{\Delta f}$.
By studying the derivative of $r$ we get that: 
$$r(\phi) \geq 2\sqrt{\frac{m-k}{k}}\sqrt{\frac{m-t}{t}} $$
So in order to have $\Delta h_{k,t}(\phi)\geq 0 $ it suffice that: 
\allowdisplaybreaks
\begin{small}
$$  \left[\sqrt{\frac{m-k}{k}}+\sqrt{\frac{m-t}{t}}\right]^2 \Delta g_{k,t} - \left(\frac{m-t}{t}- \frac{m-k}{k}\right) \Delta f \geq 0 $$
\end{small}
which is equivalent to:
\begin{small}
 $$\left[\sqrt{\frac{m-k}{k}}+\sqrt{\frac{m-t}{t}}\right]\Delta g_{k,t}- \left[\sqrt{\frac{m-t}{t}}-\sqrt{\frac{m-k}{k}}\right] \Delta f
 \geq 0$$\end{small}
We prove that $\frac{\sqrt{k(m-t)}-\sqrt{t(m-k)}}{\sqrt{k(m-t)}+\sqrt{t(m-k)}} \leq \frac{k-t}{2}$
to recover the sufficient condition.
\end{proof}

\begin{example}
For $\alpha,\beta>0$ define the distance:
\begin{align*}
    d_{\alpha,\beta}(a,A)& = \alpha |\overline{a} \cap A|+ \beta |a \cap \overline{A}|\\
    & =-(\alpha+\beta)|a \cap A| + \beta+\alpha |A| 
\end{align*}
which generalizes the Hamming distance in the same way the Tversky index \cite{tversky1977} generalizes Jaccard's. 
$d_{\alpha,\beta}$ is associated to the separable function:
$$\psi_{d_{\alpha,\beta}}(x,k)= -(\alpha+\beta)x + \beta+\alpha k = f(x)+g(k) $$
where $f(x)=-(\alpha+\beta)x + \beta$ and $g(k)=\alpha k$.
We have:
$$\Delta g_{k,t}=\frac{\alpha}{\alpha+\beta}(k-t) \Delta f$$
So for every $d_{\alpha,\beta}$ such that $\alpha \geq \beta$ we have that:
$$ \frac{\partial \mathds{E}_{\phi,d_{\alpha,\beta}}[|A_i|]}{\partial \phi} \geq 0 $$
\end{example}

\subsection{The Hamming Distance Case - Condorcet Noise Model}
The prototypical example of a separable noise is the noise associated to the Hamming distance, which is equivalent to the Condorcet-like noise model. We will prove that for this specific noise, we can express the expected size of a voter's ballot $\mathds{E}[|A_i|]$ as a linear function of her reliability parameter. This enables us to estimate this parameter directly from the actual size of the ballot, without any prior belief about the ground truth.

Formally, consider the Condorcet noise model where for each voter $i\in N$ there exists a noise parameter $p_i \in (0,1)$ 
such that:
$$P_{p_i}(a\in A_i|a=a^*)= P_{p_i}(a\notin A_i| a\neq a^*)=p_i, \forall a \in \Alt $$
and where the belonging or not of different alternatives to the voter's ballot are independent events. 
Notice that the model supposes equal error-rates for false positives and false negatives. In particular, voters who select many alternatives would \textit{ipso facto} have a low $p_i$ (since their ballots contain many false positives) which can even be below $0.5$. In fact, constraining $p_i$ to be in $(0.5,1)$ worsened the performance of our method in the experiments.
Moreover, we can easily prove that in this case, the noise model is a non-anonymous Mallows noise model with the Hamming distance and with $\phi_i= \frac{1-p_i}{p_i}$ 
(We can have $\phi_i\geq1$ since $p_i$ can be below $0.5$):
$$ P_{p_i}(A_i|a=a^*)= p_i^m \left(\frac{1-p_i}{p_i}\right)^{d_H(a^*,A_i)}, \forall a \in \Alt$$
We will show that in this particular case, we can give an explicit formula of the expected size of a voter's approval ballot as a linear function of her precision parameter $p_i$.
\begin{theorem}\label{linear_exp}
For $m\geq 2$, we have that:
$$\mathds{E}_p[|A_i|]=(m-1)-(m-2)p$$
\end{theorem}

\begin{proof}
Let $a\in \Alt$:
\begin{small}
\begin{align*}
    \mathds{E}_{p}[|A_i||a^*=a] & = \mathds{E}[\sum_{b \in \Alt} \mathds{1}\{b \in A_i\}|a^*=a] \\
    & = \sum_{b \in \Alt} P(b \in A_i|a^*=a) \\
    & = P(a \in A_i|a^*=a) + \sum_{b \neq a} P(b \in A_i|a^*=a) \\
    & = p + (m-1)(1-p)\\
    & = (m-1)-(m-2)p
\end{align*}
\end{small}
Thus we have that:
\begin{align*}
    \mathds{E}_p[|A_i|] &= \sum_{a \in \Alt} \mathds{E}_p[|A_i||a^*=a] P(a^*=a)\\
    &= \sum_{a \in \Alt} \left[(m-1)-(m-2)p\right] P(a^*=a)\\
    &= (m-1)-(m-2)p   
\end{align*}

\end{proof}
Theorem \ref{linear_exp} 
gives us a simple approach to estimate $p_i$ by maximum likelihood estimations given some observations of $A_i$ without a need to know the ground truth $a^*$.

\section{Experiments}\label{sec: experiments}
We took the three image annotation datasets, originally collected in \cite{approval2015} for incentive-design purposes\footnote{Accessible on the author's webpage: \url{https://cs.cmu.edu/~nihars/data/data_approval.zip}}, and used them to test our hypothesis and to assess the accuracy of different aggregation rules of interest. 

Each dataset consists of a set of approval profiles of a number of voters (participants) who had to select \emph{all the alternatives that they thought were correct} in a number of instances (images), namely:
\begin{compactitem}
    \item Animal task: $16$ images/questions and $6$ alternatives.
    \item Texture task: $16$ images/questions and $6$ alternatives.
    \item Language task: $25$ images/questions and $8$ alternatives.
\end{compactitem}

From now on, a dataset denotes the set $N$ of $n$ voters, $\Alt=\{a_1,\dots,a_m\}$ the set of alternatives, $L$ approval profiles $A^z=(A_1^z,\dots,A_n^z)$ each associated to an image $z$ with ground truth alternative $a_z^*$.

\subsection{Ballot Size and Reliability}
To test the hypothesis that smaller ballots are more reliable, we introduce the \emph{size-normalized accuracy} which is defined for each dataset and each $k \in [1,m-1]$ as:
$$\mbox{SNA}(k)=\frac{1}{k} \frac{|\left\{ A_i^z, |A_i^z|=k, a^z_* \in A_i^z\right\}|}{|\left\{A_i^z, |A_i^z|=k \right\}|} $$
It can be interpreted as the probability of recovering the ground truth after drawing randomly (uniformly) an alternative from a ballot of size $k$. Notice that if smaller ballots were not more reliable, one would expect that, for instance, ballots of size $2$ are twice more probable to contain the ground truth than ballots of size $1$, so the chance of finding the ground truth after randomly picking an alternative from a $2$-sized ballot is equal to the chance that a singleton ballot selects the ground truth. So we would expect that $\mbox{SNA}$ is almost constant for all $k$.

However when we compute the $\mbox{SNA}$ for the three datasets (see Figure \ref{fig: sna}) we can clearly see that it decreases for the bigger ballots,
which confirms that the alternatives selected in smaller approval ballots are more likely to coincide with the ground truth.
\begin{figure}[h]
     \centering
     \begin{subfigure}[b]{0.15\textwidth}
         \centering
         \includegraphics[width=\textwidth]{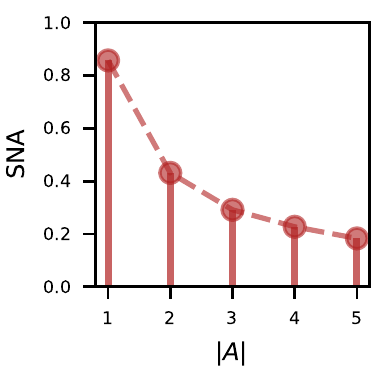}
         \subcaption{Animals}
          \label{fig: sna_animals}
     \end{subfigure}
     \begin{subfigure}[b]{0.15\textwidth}
         \centering
         \includegraphics[width=\textwidth]{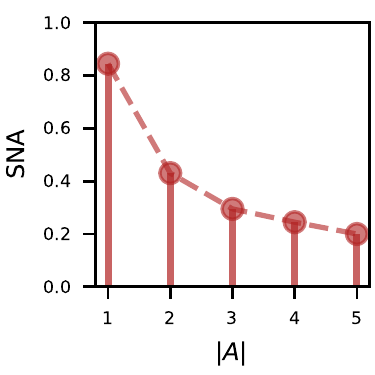}
         \subcaption{Textures}
          \label{fig: sna_textures}
     \end{subfigure}
     \begin{subfigure}[b]{0.15\textwidth}
         \centering
         \includegraphics[width=\textwidth]{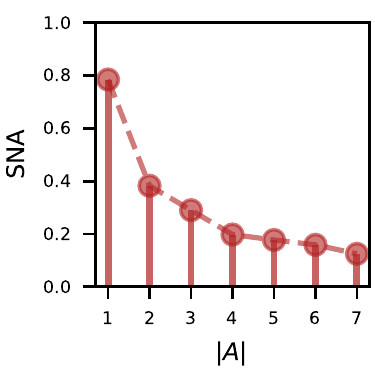}
         \subcaption{Languages}
         \label{fig: sna_languages}
     \end{subfigure}
        \caption{Size-normalized accuracy for three datasets}
        \label{fig: sna}
\end{figure}

\subsection{Aggregation}

Since we are mostly interested in the single-question wisdom of the crowd problem, we will only consider aggregation rules that operate question-wise (voters' answers on different questions do not affect the output of the rule for a given question).
We will use the following aggregation methods
(we include more methods in the Appendix)
\footnote{The code can be found at \url{https://github.com/taharallouche/Truth_Tracking-via-AV}}
:
\paragraph{Condorcet:} For each instance with approval profile $A^z=(A_1^z,\dots,A_n^z)$ and ground truth $a^*_z$, we suppose that each voter has a precision parameter $p_i^z$ such that:
$$P_{p_i^z}(a\in A_i^z|a=a^*_z)= P_{p_i^z}(a\notin A_i| a\neq a^*_z)=p_i^z, \forall a \in \Alt $$
and where the belonging or not of different alternatives to the voter's ballot are independent events.
We know that if these parameters were known, then the maximum likelihood estimation rule returns the alternative $\hat{a}_z$ such that:
$$\hat{a}_z= \argmax_{a \in \Alt} \sum_{\substack{i:a\in A_i^z }} \ln \frac{p_i^z}{1-p_i^z} $$
To estimate the parameters $p_i^z$, we will make use of the expression in Theorem \ref{linear_exp} that states that:
$$\mathds{E}_{p_i^z}[|A_i^z|]=(m-1)-(m-2)p_i^z $$ 
and set:
$$\hat{p}_i^z = \proj_{[\varepsilon,1-\varepsilon]} \frac{m-1-|A_i^z|}{m-2} $$
The projected quantity is simply the maximum likelihood estimation of $p_i^z$ with a single sample (the actual observation of the voter's ballot). We project it into a closed interval to avoid having $\hat{p}_i^z=1$ (which yields an infinite weight to the voter) whenever $|A_i^z|=1$, and $\hat{p}_i^z=0$ whenever $|A_i^z|=m-1$. So the aggregation rule finally outputs:
$$\hat{a}_z= \argmax_{a \in \Alt} \sum_{\substack{i:a\in A_i^z }} \ln \frac{\hat{p}_i^z}{1-\hat{p}_i^z} $$
which is size-decreasing.

\paragraph{Jaccard:} Here we suppose that for each instance $A^z=(A_1^z,\dots,A_n^z)$, the noise model is as follows:
$$P_{\phi,d_J}( A_i^z|a^*_z=a) =\frac{1}{\beta} \phi^{d_J(a^*_z , A_i^z)},\forall a \in \Alt$$
where $d_J(a,A)=1-\frac{|a \cap A|}{|A|+1-|a\cap A|}$.

We saw in Example \ref{ex: Jacc} that the maximum likelihood estimation rule $\zeta_{d_J}$ is a size approval rule with weights
$w_j = \nicefrac{1}{j} $.


\paragraph{Simple approval:} We will compare all these rules to the benchmark SAV rule where for each instance $A^z=(A_1^z,\dots,A_n^z)$:
$\hat{a}_z=\argmax_{a \in \Alt} |\left\{A_i^z, a \in A_i^z \right\}| $

\subsection{Results}
For each task, we took $25$ batches for each different number of voters, and applied the aforementioned rules. We measure the accuracy of each rule, outputing the esimates $\hat{a}^z $ for each instance, defined as
$$ \frac{1}{L} \sum_{z=1}^L \mathds{1}\{a^*_z = \hat{a}_z\} $$
The results are shown in figures \ref{fig: animals}, \ref{fig: textures} and \ref{fig: languages}   respectively for the Animals, Textures and Languages datasets.

\begin{figure}[h]
     \centering
     \begin{subfigure}[b]{0.44\textwidth}
         \centering
         \includegraphics[width=\textwidth]{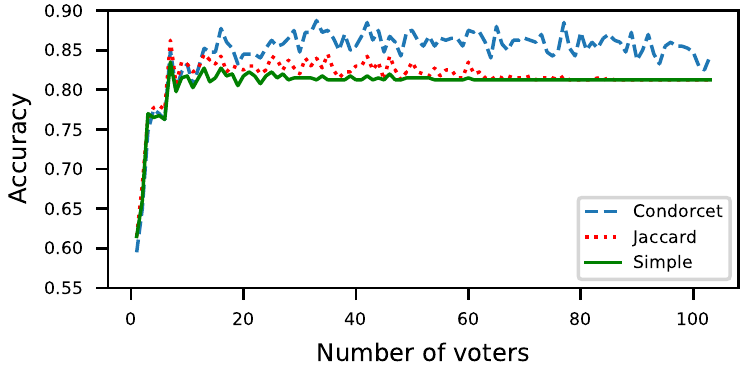}
         \subcaption{Animals}
          \label{fig: animals}
     \end{subfigure}
     \begin{subfigure}[b]{0.44\textwidth}
         \centering
         \includegraphics[width=\textwidth]{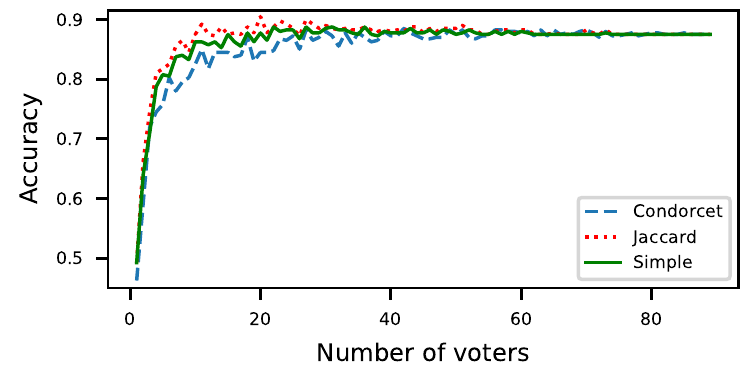}
         \subcaption{Textures}
          \label{fig: textures}
     \end{subfigure}
     \begin{subfigure}[b]{0.44\textwidth}
         \centering
         \includegraphics[width=\textwidth]{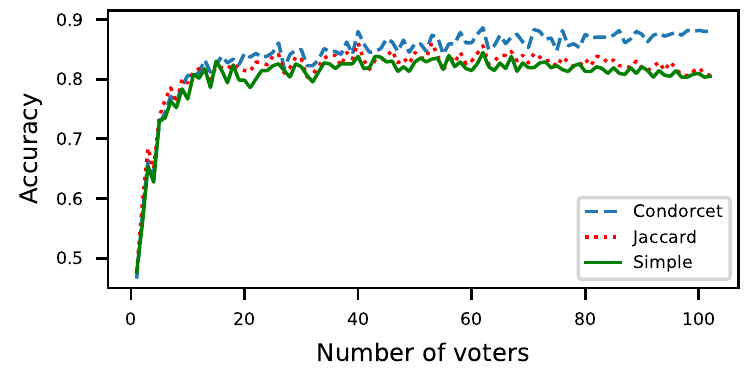}
         \subcaption{Languages}
         \label{fig: languages}
     \end{subfigure}
        \caption{Accuracy of different rules on real-life datasets}
        \label{fig: accuracy}
\end{figure}

\paragraph{Observations:}
First we notice that for all the three datasets, the aggregation rules associated to Jaccard anonymous noise show slightly better accuracy than the simple approval rule especially for small number of voters. 

We can also see that the aggregation rule associated to the non-anonymous Condorcet noise show significant improvement in the accuracy compared to this rule for Animals and Languages (specially for relatively big numbers of voters). However it fails to outperform it for the Textures dataset, where it only shows similar accuracies to the standard rule as the number of voters grows. This can be the result of the poor estimation quality which uses only one sample. 

\section{Conclusion}\label{sec: conclusion}
We propose a novel approach for epistemic approval voting based on the intuition that more reliable votes contain fewer alternatives. First, we show that for different anonymous variants of Mallows-like noise models, the maximum likelihood rule is size-decreasing, \textit{i.e} it assigns more weight to smaller ballots. Then we consider non-anonymous noises and give a sufficient condition to have an expected size of the ballot which increases as a voter gets less reliable. In particular, we prove that for a Condorcet-like noise, the expected number of approved alternatives decreases linearly with the voter's precision. Finally, we conduct experiments to test our hypothesis on real data and to assess the performances of different aggregation rules.

These methods may fail in two possible scenarios. First, if the voters do not respond truthfully. In this case, a voter can select a single alternative even though she totally ignores the correct answer. Second, if a large enough group of non-expert voters are mistakenly over-self-confident, whereas the experts are uncertain about their responses.

\bibliography{ms.bib}

\end{document}